\documentclass[11pt]{article}
\usepackage[numbers]{natbib}
\usepackage{fullpage}

\usepackage[auth-sc-lg]{authblk}

\usepackage{bm}
\usepackage{soul}
\usepackage{xspace}
\usepackage[f]{esvect}
\usepackage{amsthm}
\usepackage{amsmath}
\usepackage{amssymb}
\usepackage{mathtools}
\usepackage{mathrsfs}
\usepackage[capitalise]{cleveref}
\usepackage[shortlabels]{enumitem}
\usepackage{algorithm}
\usepackage[noend]{algpseudocode}
\usepackage{csquotes}
\usepackage{needspace}
\usepackage{xfrac}
\DeclarePairedDelimiter{\ceil}{\lceil}{\rceil}
\DeclarePairedDelimiter{\floor}{\lfloor}{\rfloor}

\providecommand{\set}[1]{\ensuremath{\{ #1 \}}\xspace}

\providecommand{\card}[1]{\ensuremath{| #1 |}\xspace}
\providecommand{\setbuild}[2]{\ensuremath{\set{#1 \mid #2}}\xspace}

\newcommand{\comp}[1]{\overline{#1}}

\newcommand{\hypen}{\text{-}}
\renewcommand{\th}{^\text{th}}
\newcommand{\Kth}[1]{\ensuremath{#1^{\text{th}}}}

\newcommand{\V}{V}
\newcommand{\C}{C}

\newcommand{\alg}{f}
\newcommand{\obj}{*}
\newcommand{\per}{\nu}

\newcommand{\election}{$\elec = (\V, \C, \prof)$\xspace}

\newcommand{\elec}{\mathcal{E}}
\newcommand{\prof}{{\vv{\succ}}}
\renewcommand{\succeq}{\succcurlyeq}
\renewcommand{\top}{\mathsf{top}}
\renewcommand{\bot}{\mathsf{bot}}
\newcommand{\cons}{\sim}

\newcommand{\apv}[1]{#1\hypen\mathsf{apv}}
\newcommand{\kapv}{\apv{k}}
\newcommand{\plu}{\mathsf{plu}}

\newcommand{\score}{\mathsf{score}}
\newcommand{\cost}{\mathsf{cost}}
\newcommand{\dist}{\mathsf{dist}}

\newcommand{\pvec}{\ensuremath{\bm{p}}\xspace}
\newcommand{\p}[1]{p(#1)}
\newcommand{\qvec}{\ensuremath{\bm{q}}\xspace}
\newcommand{\q}[1]{q(#1)}
\newcommand{\pq}{\ensuremath{{(\pvec, \qvec)}}\xspace}

\newcommand{\domg}[1]{\ensuremath{G_{#1}}\xspace} % domination graph
\newcommand{\apvveto}[1]{\ensuremath{#1\hypen\textsc{ApprovalVeto}}\xspace}
\newcommand{\kapvveto}{\apvveto{k}}

\newcommand{\pluveto}{\textsc{PluralityVeto}\xspace}
\newcommand{\plumatching}{\textsc{PluralityMatching}\xspace}

\newcommand{\votebyveto}{\textsc{VoteByVeto}\xspace}

\newcommand{\AVC}{\mathbb{AVC}}
\newcommand{\avc}[2][]{\ifthenelse{\equal{#1}{}}{\AVC_{#2}}{\AVC_{#2}(#1)}}

\newcommand{\DMC}{\mathbb{MIN}}
\newcommand{\dmc}[2][]{\ifthenelse{\equal{#1}{}}{\DMC_{#2}}{\DMC_{#2}(#1)}}

\theoremstyle{plain}
\newtheorem{theorem}{Theorem}

\newtheorem{lemma}{Lemma}
\newtheorem{corollary}{Corollary}

\theoremstyle{definition}
\newtheorem{definition}{Definition}

\newenvironment{rtheorem}[3][]{
\bigskip
\noindent \ifthenelse{\equal{#1}{}}{\bf #2 #3}{\bf #2 #3 (#1)}
\begin{it}}{\end{it}}

\title{\kapvveto: A Spectrum of Voting Rules \\ Balancing Metric Distortion and Minority Protection}

\author[ ]{Fatih Erdem Kizilkaya\thanks{Contact Author}}
\author[ ]{ David Kempe}

\affil[ ]{\textit{University of Southern California}} 
\affil[ ]{\normalsize \tt \{FatihErdemKizilkaya, David.M.Kempe\}@gmail.com}

\date{}

\begin{document}

\maketitle

\begin{abstract}
In the context of single-winner ranked-choice elections between $m$ candidates, 
we explore the tradeoff between two competing goals in every democratic system:
the \emph{majority principle} (maximizing the social welfare) and 
the \emph{minority principle} (safeguarding minority groups from overly bad outcomes).
To measure the social welfare, we use the well-established framework of \emph{metric distortion} subject to various objectives: \emph{utilitarian} (i.e., total cost), \emph{$\alpha$-percentile} (e.g., median cost for $\alpha = 1/2$), and \emph{egalitarian} (i.e., max cost).
To measure the protection of minorities, we introduce the \emph{$\ell$-mutual minority criterion},
which requires that if a sufficiently large (parametrized by $\ell$) coalition $T$ of voters ranks all candidates in $S$ lower than all other candidates, then none of the candidates in $S$ should win.
The parameter $\ell$ allows the criterion to interpolate between the minimal requirement that the winner must not be ranked \emph{last} by a~strict majority (when $\ell=2$) and
the strongest protection from bottom choices (when $\ell = m$).
The~highest $\ell$ for which the criterion is satisfied provides a well-defined measure of \emph{mutual minority protection} (ranging from 1 to $m$).

Our main contribution is the analysis of a recently proposed class of voting rules called \kapvveto,
offering a~comprehensive range of trade-offs between the two principles. 
This class spans between \pluveto (for $k=1$) --- a simple voting rule achieving optimal metric distortion --- and \votebyveto (for $k=m$) which picks a~candidate from the proportional veto core.
We show that \kapvveto has minority protection at least $k$, 
and thus, it~accommodates any desired level of minority protection.
However, this comes at the price of lower social welfare.
For the utilitarian objective, the metric distortion becomes ${2 \cdot \min (k+1, m)-1}$, i.e., increases linearly in $k$.
For the $\alpha$-percentile objective, the~metric distortion is the optimal value of~5 for $\alpha \ge k/(k+1)$ and unbounded for $\alpha < k/(k+1)$, i.e., the~range of $\alpha$ for which the rule achieves optimal distortion becomes smaller.
For the egalitarian objective, the metric distortion is the~optimal value of 3 for all values of $k$. 

\end{abstract}

\begin{displayquote}

\emph{All, too, will bear in mind this sacred principle, that though the~will of the~majority is in all cases to prevail, that will to be rightful must be reasonable; that the~minority possess their equal rights, which equal law must protect, and to violate would be oppression.}

\hfill Thomas Jefferson (in his First Inaugural Address)
\end{displayquote}

\begin{displayquote}
\emph{If a majority be united by a common interest, the rights of the minority will be insecure.}

\hfill James Madison (Federalist No. 51)
\end{displayquote}

\section{Introduction}
\label{sec:intro}

The foundation of democracy is the \emph{majority principle} --- the ideal that a decision should be based on the opinion of the majority, ensuring that the outcome is as good as possible for as many individuals as possible.
However, following this principle in and of itself poses an immediate risk, commonly known as the \emph{majority tyranny}, in which the majority pursues exclusively its own objectives at the expense of the interests of the minority factions.
Thus, constitutional democracies also incorporate the \emph{minority principle} --- the ideal that the authority of the majority should be limited to protect individuals or groups from overly bad outcomes. 
One of the main challenges in building and maintaining a democracy is to find the right balance between the contradictory factors of the majority and minority principles, 
as pointed out by the third and fourth presidents of the United States at the opening.

As a case study, consider \emph{single‐choice elections} in which each of the $n$ voters selects exactly one of the $m$ candidates on the ballot. 
In this setting, there are only two natural voting rules: 
the \emph{plurality rule}, where each voter votes for their most preferred candidate and the candidate with the most votes wins;
and the \emph{veto rule}, where each voter votes against their least preferred candidate and the candidate with the fewest (veto) votes wins. 
The plurality rule satisfies the~\emph{majority criterion} --- if a candidate $c$ is the best outcome for a majority of voters, the plurality~rule will always elect~$c$. 
Conversely, the veto rule satisfies the \emph{minority criterion} --- if a candidate $c$ is the~worst outcome for more than $n/m$ voters, the veto rule will never elect~$c$.\footnote{This is the best possible guarantee, as voters might be divided into $n/m$ groups with distinct worst outcomes.}
Thus, there is a sharp trade-off between the majority and minority principles when each voter submits a vote comprising a single candidate. 

By contrast, we~show that \emph{ranked-choice elections} --- in which voters submit a complete ranking over all candidates --- permit a spectrum of trade‐offs between the majority and minority principles.
To this end, we~present an in-depth analysis of \kapvveto \cite{GeneralizedVetoCore} --- a class of voting rules parameterized by an integer $1 \le k \le m$ which determines the trade-off between the two principles.
\kapvveto endows each voter with $k$ approval votes and $k$ veto votes, which are processed as follows:

\begin{description}[leftmargin=17pt, itemsep=0pt, topsep=0pt]
	\item[Approval Votes] First, each voter approves (their most favorite) $k$ candidates. 
	As a result, each candidate $c$ starts with a~score equal to their \emph{$k$-approval score} --- the~number of voters who have $c$ among their top $k$ choices.
	\item[Veto Votes] Then, the $nk$ veto votes are processed one by one in an~arbitrary order.  
	A veto vote of voter $v$ starts from $v$'s bottom choice (i.e., the candidate ranked lowest by $v$) among candidates that are not eliminated yet, and eliminates those whose score is 0. 	
	When the~veto vote arrives at a~candidate $c$ with positive score, it~decrements $c$'s score by 1 and terminates, that is, $c$ is not eliminated at this point even if the score is now 0.
\end{description}

\kapvveto declares candidates who are not eliminated until the~end as (tied) winners.
The~exact choice of winner(s) might differ based on the~order in which veto votes are processed.\footnote{As we briefly discuss in \cref{sec:conclusion}, one can also process the veto votes simultaneously to avoid the arbitrariness due to the choice of veto order.}
We~study the~set of all candidates who can emerge as winners for some veto order. 
We~refer to these winners as the \emph{$k$-approval veto core} due to a characterization by \citet{GeneralizedVetoCore} (given in \cref{sec:prelims}) reminiscent of definitions of ``core'' throughout social choice and game theory.
Indeed, the possible winners of $\apvveto{m}$ are exactly those in the proportional veto core defined by \citet{moulin:proportional-veto-principle}. We will discuss this core in detail below.

We are interested in how \kapvveto trades off between the majority and minority principles as the parameter $k$ is varied.
As primary motivation, we consider guarantees that can be obtained at the extreme cases $k=m$ and $k=1$.

\begin{itemize}[leftmargin=17pt, itemsep=0pt, topsep=0pt]
	\item The $m$-approval veto core is the same as the \emph{proportional veto core} (defined in \cref{sec:prelims}), which consists of all candidates that are not ``disproportionately bad'' for any coalition of voters large enough to be considered a minority group.
	Specifically, if an $\alpha$-fraction of voters prefer a $1-\alpha$ fraction of candidates over candidate~$c$, then $c$ is excluded from the proportional veto core.
	Any voting rule picking a candidate from the proportional veto core satisfies the minority criterion (discussed above in the single-choice elections setting), and thus provides the~best possible guarantee for protection of minorities from worst outcomes (i.e., bottom choices).
	
	\item At the other extreme, $\apvveto{1}$ equals \pluveto\footnote{To be precise, $\apvveto{1}$ is equivalent to an extension of \pluveto allowing for tied outcomes, introduced in \cite{GeneralizedVetoCore}; the only difference is that the original rule immediately eliminates all candidates whose score reaches 0, and thus, it insists on picking a single winner by declaring the last eliminated candidate as the winner, even for elections with obvious ties.} \cite{PluralityVeto} 
	which emerged from a~line of research towards designing a voting rule with optimal \emph{metric distortion} \citep{anshelevich:bhardwaj:elkind:postl:skowron, goel:krishnaswamy:munagala, munagala:wang:improved, gkatzelis:halpern:shah:resolving}.
	The key assumption in metric distortion is that voters and candidates are jointly embedded in an unknown metric space such that candidate $c$ is at least as close to voter $v$ as candidate $c'$ whenever $v$ prefers $c$ to~$c'$.\footnote{The motivation is that the distance between a voter and a candidate represents how much their opinions/positions on key issues differ, or simply, the \emph{cost} incurred by the voter if the candidate is elected. This generalizes the~classical notion of \emph{single-peaked} preferences \citep{black:rationale,downs:democracy,moulin:single-peak}, which considers embeddings specifically on a line instead of a~general metric space.}
	Among all such metric spaces, the worst-case ratio between the total distance of voters to a~candidate~$c$ and their total distance to an optimal candidate is referred to as the~metric distortion of~$c$ (defined formally in \cref{sec:distortion}). 
   \pluveto always returns a candidate with metric distortion at most~3, which is the~best possible guarantee.
   Given that the objective is minimizing the total distance (i.e., cost) of all voters to the winner, \pluveto is prone to majority tyranny, and thus represents an extreme case of ignoring the minority principle.
\end{itemize}

These two extreme cases achieve the best possible guarantees, respectively, for minorities and the majority, albeit under two very different frameworks. 
The former follows an \emph{axiomatic} approach while the latter follows a \emph{welfarist} approach.
In~order to understand the tradeoff as $k$ is varied, we~consider generalizations of both types of guarantees as follows.

\begin{itemize}[leftmargin=17pt, itemsep=0pt, topsep=0pt]
	\item We generalize the minority criterion by not only considering bottom choices but the full rankings.
	Our definition is inspired by the~classical notion of \emph{solid coalitions} \citep{Dummett1984}. 
	We say that a coalition $T$ of voters \emph{solidly vetoes} a subset $S$ of candidates if every voter in $T$ ranks all candidates in $S$ below all other candidates.
	We~define the \emph{mutual minority criterion} as follows: 
	if a~subset $S$ of candidates is solidly vetoed by a coalition $T$ whose size is proportionally larger than $S$ (i.e., $\card{T}/n > \card{S}/m$), then no candidate in $S$ should win (because otherwise the minority group $T$ would face a~disproportionately bad outcome).
	We introduce a parameter $\ell$ to capture varying levels of minority protection.
	Specifically, we require that $\card{T}/n > \card{S}/\ell$ (so~that $T$ needs to be larger to veto $S$ for smaller $\ell$).
	For $\ell=1$, every voting rule vacuously satisfies this criterion.
	For $\ell=2$, it coincides with the well-known \emph{majority loser criterion}, which only imposes the~minimal requirement that the winner must not be ranked last by a~strict~majority.
	Given a voting rule $f$ (or a winning candidate $c$), we~refer to the largest $\ell$ for which the criterion is satisfied as the~\emph{mutual minority protection} of $f$ (or $c$).
	(See \cref{sec:minority} for the formal definitions.)
	
	\item \pluveto, or any other rule with optimal metric distortion, is prone to majority tyranny due to the~choice of the~\emph{social cost function}.\footnote{That is, the total distance of voters to a given candidate, which is broadly referred to as the \emph{utilitarian social~cost} function in economics.}
	The framework can be adapted to use other social cost functions;
	indeed, \citet{anshelevich:bhardwaj:elkind:postl:skowron} also considered the \emph{median social cost} (i.e., the median of voters' distances to the~candidate) for this very reason ---  
	focusing on the median voter reduces the impact of outliers, i.e., voters with very high or very low costs.
	More generally, \citet{anshelevich:bhardwaj:elkind:postl:skowron} consider the \emph{$\alpha$-percentile social cost}: for~any~given $\alpha \in [0, 1)$, the~distance of the~\Kth{\floor{\alpha \cdot n + 1}} closest voter to the~candidate, which
	captures the implicit goal of protecting minorities comprising a $1-\alpha$ fraction of the population.
	The most protection to minorities --- namely, every single individual --- is provided by the~\emph{egalitarian social cost} \citep{improved_metric_distortion}, i.e., the~maximum distance of all voters to the candidate.   
	(See \cref{sec:distortion} for the formal definitions.)
\end{itemize}

\paragraph{Our Contributions}

Our main contribution is an analysis of \kapvveto following both the~axiomatic and welfarist approaches, respectively, via the mutual minority protection notion and the metric distortion framework.    

The main result of our \emph{axiomatic} analysis (presented in \cref{sec:minority}) is that every candidate in the~$k$-approval veto core (i.e., every possible winner of \kapvveto) has mutual minority protection at least $k$.
This confirms that as $k$ increases, the protection of minority groups under \kapvveto increases gradually.

We complement the axiomatic analysis with a welfarist analysis (presented in \cref{sec:distortion}) using the metric distortion framework under various objectives (social cost functions) discussed earlier.
Our main results here can be summarized as follows:

\begin{theorem} \label{thm:intro-summary-distortion}
For every $k \in \set{1, \ldots, m}$, the following results hold:
\medskip
\begin{enumerate}[1.,noitemsep,topsep=0pt,leftmargin=0.55cm]
	\item Every candidate in the $k$-approval veto core has metric distortion at most $2 \min(k+1, m) - 1$ (with respect to the utilitarian social cost, i.e., the sum of distances), and this bound is tight.
	\item For every $\alpha \ge \frac{k}{k+1}$, every candidate in the $k$-approval veto core has $\alpha$-percentile metric distortion at most 5, and this bound is tight.
	\item For every $\alpha < \frac{k}{k+1}$, there exist instances in which a candidate in the $k$-approval veto core has unbounded $\alpha$-percentile metric distortion.
	\item Every candidate in the $k$-approval veto core has egalitarian metric distortion at most 3, and this~bound is tight.
\end{enumerate}  
\end{theorem}

The first result above shows that as $k$ increases, the $k$-approval veto core sacrifices the welfare of the majority in return for higher minority protection, as captured by the axiomatic analysis.
However, the other results do not reflect this trade-off, even though the 
$\alpha$-percentile and egalitarian objectives appear to favor the minority principle.
Specifically, the~second and third results show that as $k$ increases,
the $k$-approval veto core focuses exclusively on increasingly egalitarian objectives
--- protecting minorities comprising at~most a $\frac{1}{k+1}$ fraction of the population --- 
while offering no approximation guarantees for larger groups.
To~match the upper bound of 5 for the case of $\alpha \geq \frac{k}{k+1}$, we~prove that for every $\alpha \in [1/2, 1)$, and \emph{every deterministic voting rule}, there~exist instances where the~$\alpha$-percentile distortion is arbitrarily close to~5.
This lower bound had been previously shown only for the range $\alpha \in [1/2,2/3)$; we thus close one of the~remaining gaps for metric distortion \cite{anshelevich:bhardwaj:elkind:postl:skowron}.
Lastly, the fourth result shows that one~can do better in the egalitarian setting, and the $k$-approval veto core achieves the optimal distortion of~3, for~all~$k$.  

We leave this discrepancy between our axiomatic and welfarist analyses as an open question for future work, and discuss it in more detail in \cref{sec:conclusion}.

\paragraph*{Related Work}
The notion of distortion was initially studied in the setting of normalized \emph{utilities},  i.e., each voter has non-negative utilities for candidates, adding up to 1 
\citep{BCHLPS:utilitarian:distortion,caragiannis:procaccia:voting,procaccia:approximation:gibbard,procaccia:rosenschein:distortion}.
However, without further assumptions on the~structure of the~utilities, the distortion of all voting rules can be very high \citep{caragiannisSubsetSelectionImplicit2017}.
One very natural and fruitful type of restriction was proposed by \citet{anshelevich:bhardwaj:postl,anshelevich:bhardwaj:elkind:postl:skowron}, who imposed a \emph{metric} structure. 
This metric structure is much more naturally understood when the utilities are negative, or equivalently, when we interpret the voters as having \emph{costs} for different candidates which are characterized by their distances.
The fact that distances must obey the triangle inequality restricts the structure of costs, and allows different voting rules to exhibit a rich range of different distortion values.
Among the~key results of \citet{anshelevich:bhardwaj:postl,anshelevich:bhardwaj:elkind:postl:skowron} was a (tight) bound of 5 on the metric distortion of the Copeland rule, both with respect to utilitarian and $\alpha$-percentile objectives, for any $\alpha \in [\sfrac{1}{2},1)$. 
These were (nearly) matched by a~lower bound of 3 for the utilitarian objective, and a lower bound of 5 for the $\alpha$-percentile objective 
for~${\alpha \in [\sfrac{1}{2}, \sfrac{2}{3})}$.\footnote{For $\alpha < 1/2$, the $\alpha$-percentile distortion of any deterministic voting rule can be easily seen to be unbounded.} 
For~larger $\alpha$, the established lower bound was~3.\footnote{For $\alpha \geq \frac{m-1}{m}$, \citet{anshelevich:bhardwaj:elkind:postl:skowron} also provided an upper bound of 3, which is achieved by plurality voting.}

The gap between the upper bound of 5 and the lower bound of 3 on the utilitarian metric distortion of deterministic voting rules inspired a significant thread of research work.
Initially, the ranked pairs rule was conjectured to achieve distortion~3, which was disproved by \citet{goel:krishnaswamy:munagala} who gave a lower bound of 5; the lower bound was strengthened to $\Omega(\sqrt{m})$ by \citet{DistortionDuality}.
The first improvement in the upper bound was due to \citet{munagala:wang:improved}, who achieved distortion ${2+\sqrt{5}} \approx 4.23$ using a novel asymmetric variant of the Copeland rule.
Building on the work of \citet{munagala:wang:improved} and \citet{DistortionDuality}, the gap was finally closed by \citet{gkatzelis:halpern:shah:resolving}, who showed that the voting rule \plumatching achieves distortion 3. 
The~voting rule is an exhaustive search for a candidate whose so-called ``domination graph'' (see \cref{sec:prelims}) has a~perfect matching.
Subsequently, \citet{PluralityVeto} showed that a much simpler voting~rule, called \pluveto, achieves the same guarantee of 3; 
the rule implicitly constructs a perfect matching in the domination graph witnessing the distortion guarantee, leading to a much shorter proof.
Recall that, up to some subtle tie breaking issues, \pluveto equals \apvveto{1}; the~only difference is that the~former immediately eliminates all candidates whose score reaches 0, and thus, insists on picking a~single winner.

A more general class of voting rules, based on selecting winners who have weighted bipartite matchings in a~more general domination graph, was already considered by \citet{gkatzelis:halpern:shah:resolving}, and also by \citet{PluralityVeto}.
The~connection to the proportional veto core (\cref{def:prop-veto-core}) was observed by \citet{petersNote2023} and explored in more depth by \citet{GeneralizedVetoCore}, who established --- for general weights --- the equivalence between a general version of \kapvveto, 
matchings in generalized domination graphs, and a general definition of the $(p,q)$-veto core with weights $p,q$ (\cref{def:pq-veto-core}). 
The equivalence to matchings was observed for the case of the proportional veto~core by \citet{ianovskiComputingProportionalVeto2021} who use this equivalence to compute the proportional veto core in polynomial time.
The~utilitarian metric distortion of candidates in the general $(p,q)$-veto core was investigated in more depth by \citet{learning_augmented}, motivated in part by the goal of selecting candidates with lower distortion when an oracle provides advice.\footnote{We draw significantly on the results of \citet{learning_augmented} for the analysis of utilitarian metric distortion of candidates in the $k$-approval veto core; we also utilize one of their lemmas for the $\alpha$-percentile objective.}

Various other considerations have played a role in the analysis of metric distortion.
The use of randomization in the selection of a winner can significantly improve the metric distortion: an upper bound of $3-o(1)$ and a lower bound of 2 had been known from the early work on distortion \citep{anshelevich:postl:randomized,DistortionCommunication}.
In recent breakthrough results, both the upper and lower bounds have been improved \citep{charikar:ramakrishnan:randomized-distortion, breaking_barrier}.
Several works have achieved better bounds on metric distortion when voters can communicate additional information beyond a ranking, such as (limited) information about the \emph{strengths} of their preferences \citep{limited1,limited2,improved_metric_distortion}.
More generally, the tradeoff between communication and distortion in voting rules has been considered \citep{fain:goel:munagala:prabhu:referee,mandal:procaccia:shah:woodruff,pierczynski:skowron:approval,bentert:skowron:few-candidates,DistortionCommunication}.
We refer to the surveys by \citet{anshelevich:filos-ratsikas:shah:voudouris:retrospective, anshelevich:filos-ratsikas:shah:voudouris:reading-list} for further discussion on distortion.

A sequential veto-based mechanism was first studied formally by \citet{muellerVotingVeto1978}.
\citet{moulin:proportional-veto-principle} generalized this voting rule from individuals to coalitions and studied the core of the resulting cooperative game, and thus, introduced the~proportional veto core.
The proportional veto core has found applications in windfarm location \citep{windfarm_location}, nuclear fuel disposal \citep{nuclear_fuel_disposal}, and federated learning \citep{federated_learning}, among others.
Subsequent to the introduction of the proportional veto core, \citet{moulinVotingProportionalVeto1982} proposed a rule, called \textsc{VoteByVeto}, which elects a winner from the proportional veto~core; up~to~ties, \textsc{VoteByVeto} is equivalent to \apvveto{m}. 
\citet{moulinVotingProportionalVeto1982} also studied the strategic behavior of voters under this rule.
\citet{ianovskiComputingProportionalVeto2021} (see also the expanded version \cite{ianovskiComputingProportionalVeto2023}) showed how to compute the proportional veto core in polynomial time and introduced an anonymous and neutral rule electing from the~proportional veto core.
More recently, \citet{veto_core_consistent} also studied the axiomatic properties of rules picking a candidate from the proportional veto core.

The trade-off between the majority and minority principles has also been a central topic in other works.  
Most notably, our work is closely related to that of \citet{kondratevMeasuringMajorityPower2020},  
who examined the trade-off axiomatically across a wide range of voting rules. 
They introduced the \emph{$(q, k)$-majority criterion} which requires that if more than $q$ voters rank the same $k$ candidates within their top $k$ choices  
(possibly in different orders), then the winner must be selected from among these $k$ candidates.  
A candidate satisfies the $\ell$-mutual minority criterion if and only if it~satisfies the $(q, k)$-majority criterion for all $1 \le k \le m - 1$ and $q \ge \frac{n(m - k)}{\ell}$.
Thus, their analysis (Theorem 9) shows that every candidate $c$ in the proportional veto core satisfies the mutual minority criterion, that is, $c$~achieves the optimal mutual minority protection of $m$.

\citet{similar_approach} also follow a similar approach to our work and explore a~continuous spectrum between $k$-Borda and Chamberlin-Courant voting rules in the multi-winner setting.
These two voting rules, respectively, represent the two extremes of the majority and minority principles in the~multi-winner world.
For a detailed overview of multi-winner elections, see \cite{ABCbook}.

\section{Preliminaries}
\label{sec:prelims}

An~\emph{election} \election consists of a~set $\V$ of $n$ \emph{voters}, a~set $\C$ of $m$ \emph{candidates}, and \emph{rankings} $\prof = {(\succ_v)_{v \in V}}$.
In~this notation, $\succ_v$ is the~\emph{ranking of voter $v$}, i.e., a~total order over $\C$ which represents the~preferences of $v$.
We~write $a \succ_v b$ to express that voter $v$ prefers candidate $a$ over candidate $b$; we write  ${a \succeq_v b}$ if $a = b$ or $a \succ_v b$, and say that $v$ \emph{weakly~prefers} $a$ over $b$. 
We~also extend this notation to \emph{coalitions}, i.e., non-empty subsets of voters. 
By $a \succ_T b$, we denote that every voter in the~coalition $T$ prefers $a$ over $b$; we~write $A \succ_T B$ if $a \succ_T b$ for all $a \in A$ and $b \in B$.
The~complement symbol is always used with respect to the~``obvious'' ground set, i.e., $\comp{T} = V \setminus T$ if $T \subseteq V$, and $\comp{S} = \C \setminus S$ if $S \subseteq \C$. 

The~\emph{top $k$ choices} of voter $v$, denoted by $\top_k(v)$, are the~set of $k$~candidates that $v$ prefers over all other candidates. 
We also use $\top(v)$ to denote the~\emph{top choice} of $v$, that is, $\top_1(v) = \set{\top(v)}$.
The~\emph{plurality score} of candidate~$c$, denoted by $\plu(c)$, is the~number of voters whose top choice~is~$c$.
The~\emph{$k$-approval score} of candidate~$c$, denoted by $\kapv(c)$, is the~number of voters who have $c$ among their \emph{top $k$ choices}, i.e., $\kapv(c) = {\card{\setbuild{v \in \V}{c \in \top_k(c)}}}$.
Extending this notation to sets via addition, we define $\kapv(S) = \sum_{c \in S} \kapv(c)$ for $S \subseteq \C$.
Similarly, we write $\bot_S(v)$ for the~least preferred candidate of voter $v$ among a subset of candidates $S \subseteq C$.

A~voting rule $\alg$ is a function that, given an~election $\elec$ as input, returns a~non-empty set of candidates ${\alg(\elec) \subseteq \C}$.
We~refer to $\alg(\elec)$ as the~\emph{(tied) winners} of $\elec$ under $\alg$, or simply as the~winners under~$\alg$, when $\elec$ is clear from the~context. 
  
\subsection{$k$-Approval Veto Core}

In game theory, the term ``core'' is generally used to refer to a~set of outcomes that are not ``blocked'' by any coalition.
In general games, a coalition is said to block an outcome if the members jointly prefer to deviate to another outcome.  
Collective decisions can also be thought of as outcomes (of a game played between voters) that can be blocked by coalitions that are sufficiently large, which provides a game-theoretical perspective on social choice \cite{stability_in_voting}. 

We refer to the~set of all possible winners of \kapvveto (outlined in the introduction and specified precisely in \cref{alg:kapvveto}) as the~\emph{$k$-approval veto core}, 
because they can be characterized via a~core definition (\cref{def:kapv-veto-core} below) by a result of \citet{GeneralizedVetoCore}.
Note that the choice of winner(s) for \kapvveto depends on the~given \emph{veto order} --- a~repeated sequence of voters $\sigma = (v_1, \ldots, v_{nk})$ in which each voter occurs exactly $k$ times.

\begin{algorithm}[h]
\caption{\kapvveto with veto order $(v_1, \ldots, v_{nk})$.}
\label{alg:kapvveto}
\begin{algorithmic}[1]
\State{initialize the~set of eligible winners $W = \C$} 
\State{initialize $\score(c) = \kapv(c)$ for all $c \in \C$}
\smallskip
\For{$i = 1, \ldots, nk$} \label{alg:for-start}
	\While{$\score(\bot_W(v_i)) = 0$} \label{alg:while-start}
			\State{remove $\bot_W(v_i)$ from $W$} 
	\EndWhile \label{alg:while-end}
	\State{decrement $\score(\bot_W(v_i))$ by $1$} 
\EndFor \label{alg:for-end}	
\State{\Return{$W$}}
\end{algorithmic}
\end{algorithm}  

As mentioned above, the winners under \kapvveto can be characterized by a definition of a notion of core.
This definition is a natural generalization of the \emph{proportional veto core} of \citet{moulin:proportional-veto-principle}, as discussed below.
  
\begin{definition}[$k$-Approval Veto Core \cite{GeneralizedVetoCore}] \label{def:kapv-veto-core}
	A~coalition $T \subseteq \V$ of voters \emph{$k$-blocks} candidate $w$ with \emph{witness set} $S \subseteq \C$ if and only if  $S \succ_T w$ and
	\[\frac{\card{T}}{n} > 1 - \frac{\kapv(S)}{nk}.\]
	The~\emph{$k$-approval veto core} of an election $\elec$ is the set of all candidates not $k$-blocked by any coalition, which we~denote by $\avc[\elec]{k}$.
	(We~drop $\elec$ from the~notation when it~is clear from the~context.)
\end{definition}

Another characterization of the possible winners can be obtained via \emph{$k$-domination graphs}\footnote{A more general definition than $k$-domination graphs was first given in \cite{gkatzelis:halpern:shah:resolving}. 
The definition crystallized earlier similar definitions \citep{munagala:wang:improved,DistortionDuality} into a more concise form.
1-domination graphs were the key tool in showing that \plumatching  \cite{gkatzelis:halpern:shah:resolving} and \pluveto \cite{PluralityVeto} achieve the optimal utilitarian metric distortion of 3.}
which are the main technical tool used to analyze the distortion of \kapvveto in \cref{sec:distortion}.

\begin{definition}[$k$-Domination Graphs \cite{gkatzelis:halpern:shah:resolving}] \label{def:domination_graphs}
	The \emph{$k$-domination graph} of candidate $w$ is a~bipartite graph $\domg{k}(w)$ defined between $k$~copies of each voter $v$ and $\kapv(c)$ copies of each candidate $c$. 
	The graph $\domg{k}(w)$ contains an edge between (each copy of) $v$ and (each copy of) $c$ if and only if $w \succeq_v c$.
\end{definition} 

\cref{def:kapv-veto-core} and \cref{def:domination_graphs} characterize the possible winners of \kapvveto by means of \cref{thm:equivalence} \cite{GeneralizedVetoCore}.

\begin{theorem}[\citet{GeneralizedVetoCore}, Theorems 1--3]
  \label{thm:equivalence}
	For every~$k$ and candidate $w$, the~following three statements are equivalent: 
	\begin{enumerate}[(1)]
		\item $w$ is a winner of \kapvveto for some veto order.
		\item No coalition $k$-blocks $w$, i.e., $w \in \avc{k}$.
		\item $\domg{k}(w)$ has a~perfect matching.
	\end{enumerate}
\end{theorem}

We remark that $\avc{m-1} \subseteq \avc{m}$, because any perfect matching in $\domg{m-1}(w)$ can be extended to a perfect matching in $\domg{m}(w)$ by matching the \Kth{m} copy of voter $v$ to the bottom choice of $v$. 
This construction does not extend to other $k$, because the edge from $v$ to the \Kth{k} ranked choice of $v$ may not exist in $\domg{k}(w)$ for $k < m$.
We also remark that in general, the inclusion $\avc{m-1} \subset \avc{m}$ can be strict. For example, for 3 candidates $\set{a, b, c}$ and 12 voters of whom 7 rank $a \succ b \succ c$ and 5 rank $b \succ a \succ c$, we have that $\avc{2} = \set{a}$ while $\avc{3} = \set{a, b}$.

\subsubsection{Related Notions}

The classical notion of \emph{proportional veto core} by \citet{moulin:proportional-veto-principle} (given in \cref{def:prop-veto-core}) is equivalent to the $m$-approval veto core because $\apv{m}(S) = n \cdot \card{S}$ for all $S \subseteq C$. 

\begin{definition}[Proportional Veto Core\footnote{We slightly rearrange the~definition in \citep{moulin:proportional-veto-principle} here for clarity.} 
	\citep{moulin:proportional-veto-principle}]
	\label{def:prop-veto-core}
	A coalition $T \subseteq \V$ of voters \emph{blocks} candidate $c$ with \emph{witness set} $S \subseteq \C$ if $S \succ_T c$ and \[\frac{|T|}{n} > 1 - \frac{|S|}{m}.\]
	The~\emph{proportional veto core} is the set of all candidates not blocked by any coalition.
\end{definition}

The \kapvveto voting rule, the notion of the $k$-approval veto core, and the notion of $k$-domination graphs can all be generalized further, to fractional weights and non-uniform weights not only for candidates, but also for voters.
A general definition and proof of a general version of \cref{thm:equivalence} were given by \citet{GeneralizedVetoCore}.
The key definition is the following, for any normalized vectors $\pvec$ and $\qvec$, representing weights over voters and candidates, respectively.

\begin{definition}[\pq-Veto Core \cite{GeneralizedVetoCore}]
\label{def:pq-veto-core}
   A coalition $T \subseteq \V$ of voters \emph{\pq-blocks} candidate $c$ with \emph{witness set} $S \subseteq \C$ if $S \succ_T c$ and $\p{T} > 1 - \q{S}$.
	The~\emph{\pq veto core} is the set of all candidates not \pq-blocked by any coalition.
\end{definition}

The $k$-approval veto core equals the \pq-veto core with ${\p{v} = \frac{1}{n}}$ for all $v \in \V$ and ${\q{c} = \frac{\kapv(c)}{nk}}$ for all $c \in \C$.

\section{Minority Protection}
\label{sec:minority}

In this section, we introduce the $\ell$-mutual minority criterion, which we then use to measure the opposition of minorities to a particular candidate being chosen as the winner. 
Our main result in this section is that \kapvveto provides better protection to minorities as $k$ increases.

Recall from the case study in the introduction that when voters can only communicate a single choice to the voting rule, the best possible guarantee for protecting minorities from the worst outcomes is the minority criterion --- if~a candidate $c$ is the bottom choice for more than $n/m$ voters, then $c$ should not be elected.
When full rankings are available, stronger requirements can be formulated. 
For example, if more than $2n/m$ voters have the same bottom 2 choices (possibly in different orders), then neither candidate should be elected.
This idea can be captured more generally via the~notion of \emph{solid coalitions},
popularized by \citet{Dummett1984}.

\begin{definition}[Solid Support \cite{Dummett1984}]
\label{def:solid-support}
  A~coalition $T \subseteq \V$ \emph{solidly supports} a~subset of candidates $S \subseteq \C$ if ${S \succ_T \comp{S}}$, i.e., all voters in $T$ have $S$ as their top choices (though possibly in different orders).
\end{definition}

Note that each singleton coalition $\set{v}$ solidly supports $\top(v)$;
thus, the relation of $T$ solidly supporting $S$ is an~extension of the relation of $c$ being the top choice of $v$.
Solid coalitions are used extensively in social choice theory. 
Most notably, the well-known \emph{mutual majority criterion} is a~generalization of the majority criterion via solid coalitions, considering not only the top choices but the full rankings. 

In the same vein, we define the \emph{mutual minority criterion}, generalizing the minority criterion via solid coalitions, to consider not only the bottom choices but the full rankings.
First, we extend the relation of candidate $c$ being the bottom choice of voter~$v$ by defining solidly vetoing coalitions.   

\begin{definition}[Solid Veto]
 \label{def:solid-veto} 
	A coalition $T \subseteq \V$ \emph{solidly vetoes} $S \subseteq \C$ if $\comp{S} \succ_T S$.
\end{definition}

Then, we say that a candidate $c$ satisfies the mutual minority criterion if $c \notin S$ for all $S \subseteq \C$ solidly~vetoed by a coalition $T \subseteq V$ with size $\card{T} / n > \card{S} / m$.
As discussed in the related work, this~criterion is equivalent to satisfying the $(q, k)$-majority criterion of \citet{kondratevMeasuringMajorityPower2020} for all $1 \le k \le m - 1$ and $q \ge \frac{n(m - k)}{m}$.
Thus, their analysis (Theorem 9) shows that every candidate in the proportional veto core satisfies the mutual minority criterion. 
  
However, this strong minority protection comes at a~high cost in terms of the metric distortion, a good stand-in for social welfare.
This motivates us to propose a parameterized definition of minority protection which interpolates smoothly between the strong requirements imposed by the proportional veto core and essentially no protection of the minority.
Such a definition will then allow us to study the tradeoffs between protection of minorities and welfare of the majority, quantifying and analyzing the tradeoffs between majority and minority principles. 
 
\begin{definition}[$\ell$-Mutual Minority Criterion] \label{def:droop-minority}
	A candidate $c$ satisfies the \emph{$\ell$-mutual minority criterion} if $c \notin S$ for all $S \subseteq C$ solidly vetoed by a coalition $T \subseteq V$ with size $\card{T} / n > \card{S} / \ell$. 
\end{definition} 

We write $\dmc{\ell}$ to denote the set of all candidates satisfying the $\ell$-mutual minority criterion.
Observe that $\dmc{\ell} \subseteq \dmc{\ell-1}$ for all $\ell$, that is, the $\ell$-mutual minority criterion gets harder to satisfy for larger $\ell$.
In particular, $\dmc{1} = \C$ contains all candidates, and $\dmc{m+1} = \emptyset$ (by considering ${T = \V}$ and ${S = \C}$).
Consequently, for every candidate $c$, there exists a unique ${\ell \in \set{1, \ldots, m}}$ such~that $c \in \dmc{\ell} \setminus \dmc{\ell+1}$; this gives rise to a well-defined measure, which we refer to as the~\emph{mutual minority protection} of candidate $c$.

Our main result in this section is the following theorem, showing that higher values of $k$ provide higher mutual minority protection for candidates in the $k$-approval veto core. 

\begin{theorem}
\label{thm:minority} \label{thm:minority:prop-veto-core} 
	The mutual minority protection of every candidate in the~$k$-approval veto core is at~least $k$, i.e., $\avc{k} \subseteq \dmc{k}$.   
\end{theorem}
  
\begin{proof}
  	We show that $\comp{\dmc{k}} \subseteq \comp{\avc{k}}$.
	Let~$c \notin \dmc{k}$.
	Then, there exists a coalition $T$ solidly vetoing a subset of candidates $S \ni c$ such that $\card{T}/n > \card{S} / k$.
	Now, considering $\comp{S}$,
	we observe that $\comp{S} \succ_T c$ and
	$\kapv(\comp{S}) = nk - \kapv(S) \geq nk - n\card{S}$; 
	hence, $\card{S}/{k} \geq 1 - \kapv(\comp{S})/nk$.
	Substituting this inequality, we obtain that $\card{T}/n > 1 - \kapv(\comp{S}) / nk$. 
	By \cref{def:kapv-veto-core}, this means that $T$ $k$-blocks $c$ with witness set $\comp{S}$. 
	Hence, $c \notin \avc{k}$. 
\end{proof}

The above lower bound of $k$ is not tight.  
Indeed, the mutual minority protection of every candidate in the $1$-approval veto core is at least $2$.  
This is because $\dmc{2}$ consists of exactly the candidates that are ranked last by at most $n / 2$ voters,  
and candidates that are ranked last by strictly more than $n / 2$ voters cannot possibly win under \apvveto{1} $=$ \pluveto.

\section{Metric Distortion}
\label{sec:distortion}

In this section, we present a complete analysis of the metric distortion of the $k$-approval veto core.
We begin by reviewing the relevant definitions.

%%%%%%%%%%%%%%%%%%%%%%%%%%%%%%%%%%%%%%%%%%%%%%%%%%%%%%%%%%%%%%%%%%%%%%%%%%%%%
\subsection{Framework}
%%%%%%%%%%%%%%%%%%%%%%%%%%%%%%%%%%%%%%%%%%%%%%%%%%%%%%%%%%%%%%%%%%%%%%%%%%%%%

A~\emph{metric} over a~set $S$ is a~function $d : S \times S \rightarrow \mathbb{R}_{\geq 0}$ which satisfies the~following three conditions for all ${a, b, c \in S}$: 
(1)~Positive Definiteness: $d(a,b) = 0$ if and only\footnote{Our proofs do not require the ``only if'' condition, so technically, all our results hold for pseudo-metrics, not just metrics.} if ${a=b}$; 
(2)~Symmetry: $d(a, b) = d(b, a)$; 
(3)~Triangle inequality: $d(a,b) + d(b,c) \geq d(a,c)$.
Given an~election \election, we~say that a~metric $d$ over $\V \cup \C$ is \emph{consistent} with the~rankings $\prof$, and write $d \cons \prof$, if $d(v, c) \leq d(v, c')$ for all $v \in \V$ and $c, c' \in \C$ such that $c \succ_v c'$.

The~metric distortion framework of \citet{anshelevich:bhardwaj:elkind:postl:skowron} characterizes the~quality of a~candidate $w$ (chosen as the~\emph{winner}) based on the~distances between voters and $w$.
Specifically, we study the~following three notions of social cost.
Given a~candidate $w$ and a~metric ${d \cons \prof}$, (1)~the~\emph{utilitarian social cost} of $w$ is defined as $\cost^+_d(w) = \sum_{v \in \V} d(v,w)$; 
(2)~the~\emph{$\alpha$-percentile social cost} of $w$, for a~given ${\alpha \in [0, 1)}$, is defined as $\cost^\alpha_d(w) = d(\per^\alpha_d(w), w)$ 
where $\per^\alpha_d(w)$ denotes the~$\floor{\alpha n + 1}\th$ closest voter to $w$ under $d$;
and (3)~the~\emph{egalitarian social cost} of $w$ is defined as $\cost^1_d(w) = \max_{v \in V} d(v,w)$.%
\footnote{Note that the egalitarian social cost is not subsumed by the $\alpha$-percentile social cost as the latter is not well-defined for $\alpha = 1$.}
The~metric distortion of a~candidate under utilitarian, $\alpha$-percentile, or egalitarian social cost is defined  as follows.

\begin{definition}[Metric Distortion] \label{def:distortion}
	Under the~social cost objective ${\obj \in \set{+, \alpha, 1}}$, the~\emph{(metric) distortion} of a~candidate $c$ in an election $\elec$ is the~largest possible ratio between the~social cost of $c$ and that of an optimal candidate $c^*_d$ under any metric ${d \cons \prof}$. That is,
	\[\dist^\obj_\elec(c) = \sup_{d \cons \prof} \frac{\cost^\obj_d(c)}{\cost^\obj_d(c^*_d)}.\]
	We refer to $\dist^+_\elec(c)$, $\dist^\alpha_\elec(c)$ and $\dist^1_\elec(c)$, respectively, as the \emph{utilitarian}, the~\emph{$\alpha$-percentile} and the~\emph{egalitarian} (metric) distortion of candidate~$c$.
	Extending the~notion to voting rules, we say that a~voting rule $\alg$ has distortion (at most) $\theta$ if $\dist^*_\elec(w) \leq \theta$ for every election $\elec$ and for every winner $w$ of $\elec$ under $\alg$, i.e., $\dist^*(\alg) = \max_{\elec} \max_{w \in \alg(\elec)} \dist^*_\elec(w)$.
\end{definition}

%%%%%%%%%%%%%%%%%%%%%%%%%%%%%%%%%%%%%%%%%%%%%%%%%%%%%%%%%%%%%%%%%%%%%%%%%%%%%
\subsection{Utilitarian Metric Distortion}
%%%%%%%%%%%%%%%%%%%%%%%%%%%%%%%%%%%%%%%%%%%%%%%%%%%%%%%%%%%%%%%%%%%%%%%%%%%%%

We begin by pinning down the most frequently studied notion of distortion, namely, utilitarian metric distortion, of candidates in the $k$-approval veto core, for all $k$.
The upper bound for $k < m$ follows from a recent bound by \citet{learning_augmented} for the \pq-veto core (see \cref{def:pq-veto-core}).
Our main contribution is therefore an improved bound for $k=m$ (i.e., the proportional veto core, see \cref{def:prop-veto-core}) and a matching lower bound for all $k$.

\begin{theorem} \label{thm:util}
	 The utilitarian metric distortion of every candidate in the $k$-approval veto core is at most $2 \min(k+1, m) - 1$, and this bound is tight, i.e., for all $k$, there is an election $\elec$ with $\dist^+_\elec(c) = 2 \min(k+1, m) - 1$ for some $c \in \avc[\elec]{k}$.
\end{theorem}

We begin by proving the upper bound. 
For $k < m$, as mentioned above, we can use the following result by \citet{learning_augmented}:

\begin{lemma}[\citet{learning_augmented}, Corollary 3.12]
  \label{lem:pq-core-distortion}
	The utilitarian metric distortion of every candidate in the \pq-veto core is at most 
	\[1 + \frac{2 \max_v p(v)}{\min_c \frac{q(c)}{\plu(c)}}.\]
\end{lemma}

As discussed previously, the $k$-approval veto core is the special case  ${\p{v} = \frac{1}{n}}$ for all ${v \in \V}$ and ${\q{c} = \frac{\kapv(c)}{nk}}$ for all ${c \in \C}$;
here, the bound of \cref{lem:pq-core-distortion} reduces to $1 + 2k \max_c \frac{\plu(c)}{\kapv(c)}$.
Since $\plu(c) \leq \kapv(c)$ for all ${c \in \C}$, this implies an upper bound of $2k+1$ for the~$k$-approval veto core, matching the claimed upper bound in \cref{thm:util} for $k < m$.  
While the proof follows easily from \cref{lem:pq-core-distortion}, the latter has a rather involved proof. 
In \cref{apx:flow}, we give a simpler and self-contained proof of the bound of $2k+1$ using the flow technique of \citet{DistortionDuality}.

To give an improved upper bound of $2m-1$ for $k=m$ (i.e., for the proportional veto core), we~utilize the~following lemma of \citet{anshelevich:bhardwaj:postl}: 

\begin{lemma}[\citet{anshelevich:bhardwaj:postl}, Lemma 6]
  \label{lem:anshelevich:util}
	For every pair of candidates $w \neq c^*$ and for all metrics $d$ consistent with the rankings,
	$$\frac{\cost^+_d(w)}{\cost^+_d(c^*)} \leq \frac{2n}{\card{\setbuild{v}{w \succ_v c^*}}} - 1.$$
\end{lemma}

Using this lemma, we prove the upper bound by contrapositive. 
For any candidate $c$ such that ${\card{\setbuild{v}{c \succ_v c^*}} \geq n/m}$, \cref{lem:anshelevich:util} directly implies a distortion of at most $2m-1$.
Therefore, consider a candidate $c$ such that fewer than $n/m$ voters prefer $c$ over $c^*$, so strictly more than $n-(n/m)$ voters prefer $c^*$ over $c$.
	Since $\apv{m}(c^*)=n$, this means that the coalition of all voters $v$ with $c^* \succ_v c$ $m$-blocks $c$ with witness set $\set{c^*}$.
Hence, $c$ is not in $\avc{m}$.

We now complete the proof of \cref{thm:util} by giving matching lower bounds for all $k$ in \cref{lem:util:lower-bound}.

\begin{lemma} \label{lem:util:lower-bound}
   For all $k$ and $\epsilon > 0$, there is an election $\elec$ with $\dist^+_\elec(c) \geq 2 \min(k+1, m) - 1 - \epsilon$ for some $c \in \avc[\elec]{k}$.
\end{lemma}

\begin{proof}
	As discussed previously, $\avc[\elec]{m-1} \subseteq \avc[\elec]{m}$ for all elections $\elec$, and the claimed lower bound reduces to $2k+1$ for $k < m$. 
	Hence, if the lemma holds for all $k < m$, then it also holds for $k = m$ as $2(m-1)+1=2m-1$.
	Thus, we fix a $k < m$ without loss of generality.
	
	Let \election be an election with $\V = \set{v_1, \ldots, v_{k+1}}$ and $\C = A \cup B$ where $A = \set{c_1, \ldots, c_{k+1}}$ and $B = \set{c_{k+2}, \ldots, c_m}$. 
  	The rankings $\prof$ are such that $A \succ_V B$ and $\bot_{A}(v_i) = c_i$ for each voter $v_i \in V$.
  	Note that $\kapv(a) = k$ for all $a \in A$; and $\kapv(b) = 0$ for all $b \in B$.
  	We~show that $A \subseteq \avc{k}$ by using the characterization of $\avc{k}$ in terms of perfect matchings of $\domg{k}(c)$ (\cref{thm:equivalence}).
  	For every $i, j \leq k+1$, notice that $c_i \succeq_{v_j} c_j$.
  	Thus, $\domg{k}(c_i)$ has the perfect matching in which the~$k$~copies of $v_j$ are matched with the~$k$~copies of $c_j$, for all $i, j \leq k+1$.
  	 
 	We will now construct a metric $d$ under which the candidates $c_i$ for $i < k+1$ (particularly~$c_1$) have much higher cost than $c_{k+1}$.
	Under the metric, we have ${d(v_{k+1}, c_{k+1}) = 1+\delta}$ and $d(v_i, c_{k+1}) = 2\delta$ for all $i \neq k+1$. 
	For all $i < k+1$, we have $d(v_{k+1},c_i) = 1$, $d(v_i, c_i) = 2+\delta$, and $d(v_j, c_i) = 2$ for $j \notin \set{i, k+1}$.\footnote{To avoid relying on tie breaking, as with all our other lower bounds, one could slightly perturb distances.} 
	It~is~easy to verify that these distances form a metric consistent with $\prof$.
	The cost of $c_{k+1}$ is at~most $1+2n\delta$, while the cost of all $c_i$ with $i < k+1$ (note that at least $i=1$ satisfies this inequality for all $k$) have cost at least $2k+1$. With $\delta \to 0$, the distortion, lower-bounded by $\frac{2k+1}{1+2n\delta}$ will get arbitrarily close to $2k+1$, completing the proof.
\end{proof}

%%%%%%%%%%%%%%%%%%%%%%%%%%%%%%%%%%%%%%%%%%%%%%%%%%%%%%%%%%%%%%%%%%%%%%%%%%%%%
\subsection*{$\alpha$-Percentile Metric Distortion}
%%%%%%%%%%%%%%%%%%%%%%%%%%%%%%%%%%%%%%%%%%%%%%%%%%%%%%%%%%%%%%%%%%%%%%%%%%%%%

In this section, we show that the $\alpha$-percentile distortion of every candidate in the $k$-approval veto core is 
at most 5 for ${\alpha \ge k/(k+1)}$ (\cref{thm:perc:kavc}),
and unbounded for ${\alpha < k/(k+1)}$ (\cref{thm:perc:unbounded}).
In particular, \pluveto has $\alpha$-percentile distortion at most 5 for all ${\alpha \ge \sfrac{1}{2}}$, 
which we show to be the~best possible bound for any (deterministic) voting rule (\cref{thm:perc:deterministic}).     
Previously, \citet{anshelevich:bhardwaj:elkind:postl:skowron} had shown that the~$\alpha$-percentile distortion of every voting rule is unbounded for ${\alpha \in [0, \sfrac{1}{2})}$, at~least~5 for ${\alpha \in [\sfrac{1}{2}, \sfrac{2}{3})}$, and at~least~3 for ${\alpha \in [\sfrac{2}{3}, 1)}$.
Thus, we improve their lower bound from 3 to 5 for ${\alpha \in [\sfrac{2}{3}, 1)}$.
This establishes that \pluveto is an~optimal voting rule in terms of $\alpha$-percentile distortion for all~$\alpha$.
The only other rule that is known to achieve constant $\alpha$-percentile distortion is the Copeland rule, which enjoys the same bound of $5$ for all $\alpha \ge \sfrac{1}{2}$.
Our analysis uses the~following two lemmas:

\begin{lemma}[\citet{anshelevich:bhardwaj:elkind:postl:skowron}, Lemma 29]
\label{lem:anshelevich:perc}
   For every $\alpha \in [0, 1)$, and pair of candidates $c$~and~$c'$, 
   $\cost^\alpha_d (c) \leq \cost^\alpha_d (c') + d(c, c')$.
\end{lemma}

\begin{lemma}[\citet{learning_augmented}, Lemma 3.9]
\label{lem:berger}
For every pair of voters $v$ and $v'$, and every pair of candidates $c$ and $c'$, if~${c \succeq_v \top(v')}$, then for every metric consistent with the rankings,
\[d(v, c') + d(v', c') \ge \frac{d(c, c')}{2}.\]
\end{lemma}

\begin{theorem} \label{thm:perc:kavc}
  The $\alpha$-percentile distortion of every candidate in the $k$-approval veto core is at most~$5$ for all $\alpha \geq \frac{k}{k+1}$.
\end{theorem}

\begin{proof}
	Fix some $k \in \set{1, \ldots, m}$ and $\alpha \geq k / (k+1)$. 
	Let \election be an election, and let~$d \cons \prof$ be a~metric consistent with the rankings. 
	Let $w \in \avc{k}$, and let $c^*$ be an optimal candidate under~$d$.
	We distinguish two cases, based on whether $w$ and $c^*$ are ``close'' to each other (compared to the cost of $w$) or not.
	
	If $d(w,c^*) \leq \frac{4}{5} \cdot \cost^{\alpha}_d(w)$, then using \cref{lem:anshelevich:perc}, we obtain 
	$$\cost^{\alpha}_d(c^*) \geq \cost^{\alpha}_d(w) - d(w,c^*)
	\geq \cost^{\alpha}_d(w) - \frac{4}{5} \cdot \cost^{\alpha}_d(w)
	= \frac{1}{5} \cdot \cost^{\alpha}_d(w).$$   
	
	If $\cost^\alpha_d(w) < \frac{5}{4} \cdot d(w, c^*)$, 
	we show that there are at least $n/(k+1)$ voters $v$ with distance $d(v,c^*) \geq d(w,c^*) / 4$, which implies that $\cost^\alpha_d(c^*) \ge d(w,c^*) / 4$.
  	Assume for contradiction that this is not the case, i.e., the set $V' = \setbuild{v \in V}{d(v,c^*) < d(w,c^*) / 4}$ contains strictly more than $nk / (k+1)$ voters.

	Consider the set of candidates $C' = \setbuild{\top(v)}{v \in V'}$ which are the top choice of at least one voter in~$V'$.
	Then, the total plurality score $P := \sum_{c \in C'} \plu(c) \geq |V'| > nk / (k+1)$, and the bipartite graph $\domg{k}(w)$ contains at least $P$ copies of candidates in $C'$ (possibly more, by considering lower rankings when $k > 1$).
	Furthermore, for any $v, v' \in V'$, we have that $d(v,c^*) + d(v',c^*) < d(w,c^*) / 2$. Then, \cref{lem:berger} implies that no voter $v \in V'$ prefers $w$ over $\top(v')$ for any $v' \in V'$.
	As a result, $\domg{k}(w)$ cannot contain any edges from copies of $v \in V'$ to any copies of candidates $c \in C'$.
	In other words, the only edges to copies of candidates $c \in C'$ can come from copies of voters $v \notin V'$;
	and since there~are strictly fewer than $n/(k+1)$ voters not in $V'$, there~are strictly fewer than $nk/(k+1)$ such copies.
	On the other hand, we argued above that $\domg{k}(w)$ contains at least $P > nk/(k+1)$ copies of candidates in $C'$.
	Thus, we have exhibited a~set of $P$ nodes in $\domg{k}(w)$ whose neighborhood contains strictly fewer than $P$ nodes.
	By Hall's marriage theorem, $\domg{k}(w)$ does not have a perfect matching, contradicting the assumption of the theorem that $w \in \avc{k}$.
\end{proof}

We next show matching lower bounds.
First, we show that for $\alpha < \frac{k}{k+1}$, the $\alpha$-percentile distortion of candidates in the $k$-approval veto core may be unbounded.

\begin{theorem} \label{thm:perc:unbounded}
  For all $\alpha < \frac{k}{k+1}$, there exists an election $\elec$ such that $\dist^\alpha_\elec(w) = \infty$ for all candidates $w \in \avc[\elec]{k}$.
\end{theorem} 

\begin{proof}
	\citet{anshelevich:bhardwaj:elkind:postl:skowron} already established that the $\alpha$-percentile metric distortion of every (deterministic) voting rule is unbounded for ${\alpha \in [0, 1/2)}$.
	We therefore focus on $\alpha \in [1/2, k/(k+1))$ and thus on $k \in \set{2, \ldots, m}$.
   
   	Let $n$ be an odd number such that $\alpha n$ is not an integer, and $\ceil{\alpha n} < nk/(k+1)$. 
   	Such~an~$n$~exists; in particular, the property $\ceil{\alpha n} < nk/(k+1)$ is satisfied for all sufficiently large $n$, due to the assumption that $\alpha$ is strictly smaller than $k/(k+1)$. 
   	Let $\ell = \ceil{\alpha n} = \floor{\alpha n + 1}$ (by the assumption that $\alpha n$ is not an integer). 
   	Note that since $\alpha \geq 1/2$ and $n$ is odd, we have that $\ell > n-\ell$.

   	Let \election be an election with $k+1$ candidates $c_1, \ldots, c_k, c^*$ and $n$ voters.
   	The voters are partitioned into two sets $T \cup \comp{T} = \V$ of sizes $n-\ell$ and $\ell$.
   	Every voter $v \in T$ has $c^*$ as their bottom choice, and every voter $v \in \comp{T}$ has $c^*$ as their top choice.
   	As~a~result, the $k$-approval score of $S = \set{c_1, \ldots, c_k}$ is $nk - \ell$; the only top-$k$ votes not going to $S$ are the first-place votes of $v \in \comp{T}$ (which go to $c^*$).
   	Then, we have that
   	\[1 - \frac{\kapv(S)}{nk}
		= \frac{\ell}{nk}
		= \frac{\ceil{\alpha n}}{nk} 
		\leq \frac{1}{k+1}
		= \frac{n - \frac{nk}{k+1}}{n}
		< \frac{n-\ceil{\alpha n}}{n}
		= \frac{n-\ell}{n}
		= \frac{\card{T}}{n}.\]	
	Since $S \succ_T c^*$, the coalition $T$ $k$-blocks candidate $c^*$ with witness set $S$, and thus, $c^* \notin \avc[\elec]{k}$.
	Consequently, we~get that $\avc[\elec]{k} \subseteq S$.
	
	\needspace{\baselineskip}
	Now, consider the following metric $d$ consistent with the rankings. 
  	All voters $v \in T$ have distance $d(v, c^*) = 1$ and $d(v, c_i) = 0$ for all $c_i \in S$.
  	All voters $v \in \comp{T}$ have distance $d(v, c^*) = 0$ and $d(v, c_i) = 1$ for all $c_i \in S$.
  	Since $\card{\comp{T}} = \floor{\alpha n + 1}$, we~have $\cost^{\alpha}_d (c^*) = 0$.
  	Since $\card{T} = n-\ell < \ell = \floor{n \alpha + 1}$,
  	we~have $\cost^{\alpha}_d(c_i) = 1$ for all $c_i \in S$.
	Because there exists a~candidate (namely, $c^*$) of cost~0, yet all candidates in $S$ (and hence in $\avc{k}$) have cost 1, this completes the proof.
\end{proof}  

Next, we show that no deterministic voting rule can achieve $\alpha$-percentile distortion smaller than 5, for any $\alpha < 1$.
The proof builds on a straightforward extension of the construction by \citet{anshelevich:bhardwaj:elkind:postl:skowron},  
and is therefore deferred to \cref{apx:tight_lower_bound}.

\begin{theorem}
  \label{thm:perc:deterministic}
  For every (deterministic) voting rule $\alg$, and constants $\alpha \in [1/2,1)$ and $\epsilon > 0$, 
  there~exists an~election~$\elec$ and candidate ${w \in \alg(\elec)}$ such that $\dist^\alpha_\elec(w) \geq 5 - \epsilon$.
\end{theorem}

%%%%%%%%%%%%%%%%%%%%%%%%%%%%%%%%%%%%%%%%%%%%%%%%%%%%%%%%%%%%%%%%%%%%%%%%%%%%%
\subsection*{Egalitarian Metric Distortion}
%%%%%%%%%%%%%%%%%%%%%%%%%%%%%%%%%%%%%%%%%%%%%%%%%%%%%%%%%%%%%%%%%%%%%%%%%%%%%

Here, we show that all candidates in the $k$-approval veto core, for all~$k$, have egalitarian (metric) distortion at most~3, which is known to be the best possible guarantee \cite{improved_metric_distortion}.
This~is based on the~key observation that very minimal conditions are enough to ensure egalitarian distortion at most~3; 
in particular, all Pareto efficient candidates have egalitarian distortion at most~3.
We~begin by recalling the definition of Pareto domination:

\begin{definition}[Pareto domination]
  \label{def:pareto}
     A candidate $c$ is \emph{Pareto dominated} by a candidate $c'$ if ${c' \succ_v c}$ for every voter $v$. 
     If a candidate $c$ is not Pareto dominated by any candidate $c'$, then $c$ is \emph{Pareto efficient}.
\end{definition} 
 
We first adapt the proof of Theorem~30 of \citet{anshelevich:bhardwaj:elkind:postl:skowron}, and show that if candidate $c$ is not Pareto dominated by candidate~$c'$, the egalitarian distortion of $c$ cannot be much higher than that of~$c'$.

\begin{lemma}
  \label{lem:non-domination}
  Let \election be an election and $d \cons \prof$ be a metric.
  If $c$ and $c'$ are candidates such that $c'$ does not Pareto dominate $c$, then $\cost^1_d(c) \leq 3 \cdot \cost^1_d(c')$.
\end{lemma}

\begin{proof}
  The proof is similar to that of \cref{thm:perc:kavc}; the focus on the worst-off voter allows us to obtain sharper bounds.
  To avoid notational overload, we will write $\cost(c) := \cost^1_d(c)$ for this proof (and similarly for $c'$).
  We distinguish two cases, based on whether $c$ and $c'$ are ``close'' to each other (compared to the cost of $c$) or not.
  
  \begin{enumerate}
    \item If $c$ and $c'$ are ``close'', in the sense that $d(c,c') \leq \frac{2}{3} \cdot \cost(c)$, then using \cref{lem:anshelevich:perc}, we obtain
    	$$\cost(c') \geq \cost(c) - d(c,c') \geq \cost(c) - \frac{2}{3} \cdot \cost(c) = \frac{1}{3} \cdot \cost(c).$$
    \item If $c$ and $c'$ are not close, so $d(c,c') > \frac{2}{3} \cdot \cost(c)$,
      then let $v$ be any voter with $c \succ_v c'$; such a voter is guaranteed to exist because we assumed that $c'$ does not Pareto-dominate $c$.
      Using the triangle inequality and the fact that $v$ prefers $c$ over $c'$, as well as the fact that the furthest candidate from $c'$ is at least as far as $v$, we obtain that
      $$d(c,c') \leq d(v,c) + d(v,c')
        \leq 2 d(v,c')
        \leq 2 \cdot \cost(c').$$
      Hence, we obtain that $\cost(c) \leq \frac{3}{2} \cdot d(c,c') \leq 3 \cdot \cost(c')$. \qedhere
  \end{enumerate}
\end{proof}

Because any Pareto efficient candidate $c$ satisfies the condition of \cref{lem:non-domination} for the \emph{optimal} candidate $c^* \ (= c')$, we~immediately obtain the following corollary:

\begin{corollary}
  \label{cor:egal:pareto}
  Every Pareto efficient candidate $c$ has egalitarian metric distortion at most 3.
\end{corollary}

Every candidate in the proportional veto core is Pareto efficient because the grand coalition (of all voters) $m$-blocks every Pareto dominated candidate:
this shows that the egalitarian distortion of every candidate in the proportional veto core is at most 3.

In extending this result to the $k$-approval veto core for $k < m$, we face the obstacle that the $k$-approval veto core can contain Pareto-dominated candidates. 
This was observed for the case $k=1$ by \citet{GeneralizedVetoCore} who showed that such domination was only possible in a very limited sense: both the dominated and dominating candidate had to have plurality score 0.
Next, we extend this insight to the $k$-approval veto core in the following lemma.

\begin{lemma}
  \label{lem:limited-Pareto-domination}
  Let $\elec$ be an election, and $c \in \avc[\elec]{k}$ such that $c$ is Pareto-dominated by $c'$.
  Then, $\kapv(c') = 0$.
\end{lemma}

\begin{proof}
  We prove that if $c \in \avc{k}$ and $\kapv(c') > 0$, then $c'$ cannot Pareto-dominate $c$.
  To see this, recall from \cref{thm:equivalence} that $c \in \avc{k}$ if and only if the bipartite graph $\domg{k}(c)$ has a perfect matching $M$.
  $\domg{k}(c)$ contains $\kapv(c') > 0$ copies of $c'$, which must be matched in $M$.
  By definition of the edges of $\domg{k}(c)$, any voter $v$ matched to a copy of $c'$ must satisfy $c \succ_v c'$, and the existence of such a voter $v$ implies that $c'$ cannot Pareto-dominate $c$.
\end{proof}

We are now ready to state and prove our main theorem on the egalitarian distortion of the $k$-approval veto core.

\begin{theorem} \label{thm:egal:kavc}
  For every $k$, the egalitarian distortion of every candidate in the $k$-approval veto core is at most 3.
\end{theorem}

\begin{proof}
  Fix an election \election and a~metric $d \cons \prof$, and let $w$ be a candidate in \mbox{the $k$-approval} veto core.
  Let $c^*$ be an optimal candidate under $d$.
  If $c^*$ does not Pareto-dominate $w$, then the~result follows immediately by applying \cref{lem:non-domination} to $w$ and $c^*$.
     
  Otherwise, by \cref{lem:limited-Pareto-domination}, we first obtain that $\kapv(c^*) = 0$.
  By \cref{thm:equivalence}, there is a perfect matching $M$ of $\domg{k}(w)$.
  Let $v$ be a most distant voter from $w$ under $d$.
        $\domg{k}(w)$ must contain $k$ copies of $v$, all of which are matched under $M$ (possibly to multiple copies of the~same candidate).
  Let~$c$ be any candidate such that at least one copy of $v$ is matched to at least one copy of $c$.
  Since $w \succ_v c$ by definition of the edges of $\domg{k}(w)$ and $\cost^1_d(c) \geq d(c,v)$ by definition of egalitarian cost, 
  we~can bound $\cost^1_d(w) = d(w, v)  \leq d(c, v)  \leq \cost^1_d(c)$.

  Next, we observe that because $\domg{k}(w)$ contained a copy of $c$, and the number of copies of $c$ is $\kapv(c)$, we get that $\kapv(c) > 0$, whereas $\kapv(c^*) = 0$. 
  Therefore, at least one voter (ranking $c$ in the top $k$ positions) prefers $c$ over $c^*$; in particular, $c^*$ cannot Pareto-dominate $c$.
  By \cref{lem:non-domination}, this implies that $\cost^1_d(c) \leq 3 \cdot \cost^1_d(c^*)$, and hence $\cost^1_d(w) \leq 3 \cdot \cost^1_d(c^*)$.
\end{proof}

\needspace{5\baselineskip}
\section{Conclusion and Future Directions}
\label{sec:conclusion}

Our analysis shows that as $k$ increases, \kapvveto sacrifices welfare gradually to enhance minority protection.
Along with its simplicity, this makes \kapvveto potentially practical for settings where it is desirable to balance the majority and minority principles.
Hence, studying other axiomatic properties of \kapvveto would be of interest.
For instance, the rule violates the essential axiom of \emph{anonymity} (i.e., the requirement that all voters be treated equally a priori). 
However, by processing veto votes simultaneously, as shown in \cite{GeneralizedVetoCore}, both anonymity and \emph{neutrality} (i.e., the counterpart of anonymity for candidates) can be satisfied.

Another practical aspect of \kapvveto is that the parameter $k$ (i.e., the number of approval~votes) provides an intuitive means to adjust the desired level of mutual minority protection, albeit at the cost of some social welfare.
Perhaps the most significant question about this trade-off is whether \kapvveto achieves the optimal balance, i.e., the minimal loss in welfare to reach the desired level of mutual minority protection. 
We~leave determining the best achievable metric distortion by a voting rule satisfying the $\ell$-mutual minority criterion as an open question.  
Exploring the mutual minority protection of other voting rules could be a valuable step towards this challenging goal.
On a related note, it would also be of interest to study mutual minority protection in comparison with other notions such as \emph{normalized} distortion \citep{BCHLPS:utilitarian:distortion,caragiannis:procaccia:voting,procaccia:approximation:gibbard,procaccia:rosenschein:distortion}, and/or in more general settings such as \emph{randomized} voting \citep{anshelevich:postl:randomized,gross:anshelevich:xia:agree, charikar:ramakrishnan:randomized-distortion,breaking_barrier}.

Even though the $\alpha$-percentile and egalitarian objectives appear to favor the minority principle,  
we showed that the metric distortion of \kapvveto with respect to these objectives does not improve as $k$ increases,  
revealing a discrepancy between our axiomatic and welfarist analyses.
While resolving the optimal metric distortion conjecture for the utilitarian objective required novel approaches developed over an extended period, we showed that achieving optimal metric distortion for the egalitarian objective is straightforward: simply select a Pareto optimal candidate.
Moreover, although not all candidates in the $k$-approval veto core are Pareto optimal, their egalitarian metric distortion remains optimal.  
This suggests that formulating the minority principle within the metric distortion framework may require a more nuanced objective or entirely new ideas.

A strong assumption in our $\ell$-mutual minority criterion is that minority groups are modeled as coalitions solidly vetoing a subset of candidates $S$, i.e., all members of the coalition rank candidates in $S$ at the bottom (in some order).
However, this assumption is often unrealistic in practice, as minorities rarely form \emph{perfectly} solid coalitions.
Hence, an important direction for future work is to extend the $\ell$-mutual minority criterion to accommodate more robust models of minorities.
This direction is in parallel with the work of \citet{robust_and_verifiable} exploring robust and verifiable proportionality axioms (such as EJR+ and PJR+) in the multi-winner voting setting.
In~connection with this work, it~would also be of interest to investigate the hardness of verifying whether a given candidate satisfies the $\ell$-mutual minority criterion.

\bibliographystyle{plainnat}
\bibliography{davids-bibliography/names,davids-bibliography/conferences,davids-bibliography/bibliography,davids-bibliography/publications,EC_references,new_references}

\appendix
\section{Alternative Proof Based on Flow Networks}
\label{apx:flow}

Here, we give an alternative and direct proof of the upper bound of $2k+1$ (from \cref{sec:distortion}) on the utilitarian metric distortion of candidates in the~$k$-approval veto core.
The proof uses the flow technique of \citet{DistortionDuality} which is encapsulated in \cref{lem:dual-flow} below.

\begin{lemma}[\citet{DistortionDuality}, Lemma 3.1]
\label{lem:dual-flow}
	Given an election \election, let $H = (\V \times \C, E)$ be a~directed graph with edges defined as follows:
	\begin{itemize}
		\item Each voter $v$ has a directed \emph{preference edge} $(v,c) \to (v,c')$ for all pairs of candidates $c \succ_v c'$.
		\item Each candidate $c$ has a bi-directed \emph{sideways edge} $(v,c) \leftrightarrow (v',c)$ for all pairs of voters $v \neq v'$.
	\end{itemize}
	
	Given a pair of candidates $w$ and $c^*$, a \emph{$(w, c^*)$-flow} is a circulation $f$ (i.e., non-negative and conserving flow) on $H$
	such that exactly one unit of flow originates at node $(v,w)$ for all $v \in \V$, and flow is only absorbed at nodes $(v, c^*)$. 
	
	The \emph{cost of $f$ at voter $v$} is the total amount of flow absorbed at $(v, c^*)$, plus the total flow on sideways edges into or out of nodes $(v, c)$, for any candidate $c$, i.e.,
          \[ \cost_v(f) = \sum_{e \text{ into } (v, c^*)} f_e
            + \sum_{c \neq c^*} \sum_{v' \neq v} f_{(v', c) \to (v, c)}  + f_{(v, c) \to (v', c)}.
          \]
	The \emph{cost of $f$} is $\cost(f) = \max_{v \in \V} \cost_v(f)$.
	
	\bigskip

	This notion of cost can be used to give an upper bound on the utilitarian metric distortion of a candidate $w$ as follows:
	If for all $c^* \in \C$, there exists a $(w, c^*)$-flow $f$ with $\cost(f) \le \lambda$, then $\dist^+_\elec(w) \le \lambda$, that is, the~utilitarian metric distortion of candidate $w$ is at most $\lambda$.
\end{lemma}

We now use this lemma to get a self-contained proof for the following theorem.  

\begin{theorem} \label{thm:alternative}
	The utilitarian metric distortion of every candidate in the $k$-approval veto core is at~most $2k+1$. 
\end{theorem}

\begin{proof}
	Let $w \in \avc{k}$. 
	By \cref{thm:equivalence}, there is a~perfect matching $M$ of $\domg{k}(w)$.
	Let $M_i(v)$ denote the candidate that the \Kth{i} copy of voter $v$ is matched to in $M$.
	We will now argue that there are $k$ perfect matchings $N_i: V \to V$ between voters and \emph{voters} (for $i=1, \ldots, k$) such that $M_i(v) \in \top_k(N_i(v))$ for all $v \in V$ and $i \in \set{1, \ldots, k}$.
	Since each candidate $c$ has $\kapv(c)$ copies in~$\domg{k}(w)$, 
	we have 
	$$\sum_{i=1}^k \card{\setbuild{v \in V}{M_i(v) = c}} = \kapv(c) = \card{\setbuild{v \in V}{c \in \top_k(v)}}.$$
	
	We can therefore define a perfect matching between the $\kapv(c)$ copies of voters $v$ that are matched to $c$ under $M$ and the $\kapv(c)$ copies of voters $v'$ who have $c \in \top_k(v')$. 
	Since this matching has exactly $k$ copies of each voter in the first set and $k$ copies of each voter in the~second~set, it~defines a $k$-regular bipartite graph when the copies of nodes are contracted into a single node each. 
	Therefore, it can be partitioned into $k$ perfect matchings 
	$N_i : V \rightarrow V$ such that $M_i(v) \in \top_k(N_i(v))$ for all $v \in V$ and $i \in \set{1, \ldots, k}$, as desired.
	
	For an arbitrary candidate $c^*$, we will now describe a $(w, c^*)$-flow $f$ with $\cost(f) \leq 2k+1$, which is sufficient to prove \cref{thm:alternative} due to \cref{lem:dual-flow}.
	Recall that for each voter $v$, we have to specify how to route one unit of flow originating at $(v,w)$ to (a combination of) nodes $(v',c^*)$.
	For each voter $v$ and $i \in \set{1, \ldots, k}$, if $M_i(v) = \top(N_i(v))$ or $c^* \notin \top_k(N_i(v))$, then exactly $1/k$ units of flow originating at ${(v, w)}$ are sent along the path $(v,w) \to (v, M_i(v)) \to (N_i(v), M_i(v)) \to (N_i(v), c^*)$. 
	To argue that this is a valid path in $H$, 
	note that the first edge exists because $w \succeq_v M_i(v)$ by \cref{def:domination_graphs}.
	The second edge is a sideways edge of $M_i(v)$.
	The third edge exists because $M_i(v) \succ_{N_i(v)} c^*$ due to either of the two conditions $M_i(v) = \top(N_i(v))$ or $c^* \notin \top_k(N_i(v))$, using that $M_i(v) \in \top_k(N_i(v))$ in the second case.  

	Before we describe how the rest of the flow is routed, we first calculate the total amount of flow routed so far. 
	Focus on one voter $v$ and the flow into nodes $(v,c)$ along sideways edges. 
	By~the~construction above, such flow occurs exactly when the flow originates at $(w,v')$ with ${N_i(v') = v}$ and we~have either $c^* \notin \top_k(v)$ or $M_i(v') = \top(v)$.
	In~the~former~case, for each $c \in \top_k(v)$, there is one node $v'$ with $M_i(v') = c$ and $N_i(v') = v$, resulting in $1/k$ units of flow into $(v,c)$.
	In~the~latte~case, the node $(v,\top(v))$ receives $1/k$ units of flow from the unique $v'$ with $M_i(v') = \top(v)$ (for~some~$i$), but no other node $(v,c)$ receives any flow from sideway edges.

	Thus, the total flow routed according to the previous description is
    
    \begin{align*}
    	\card{\setbuild{v}{c^* \notin \top_k(v)}}
    	+ \frac{1}{k} \cdot \card{\setbuild{v}{c^* \in \top_k(v)}}
    	&= n - \kapv(c^*) + \frac{1}{k} \cdot \kapv(c^*) \\
		&= n - \left(\frac{k-1}{k}\right) \cdot \kapv(c^*).  
    \end{align*}
    
	In order to describe how the remaining $\frac{k-1}{k} \cdot \kapv(c^*)$ units of flow are routed, we define $V^* = {\setbuild{v \in V}{w \succeq_{v} c^*}}$ to be the set of voters $v$ who prefer $w$ over $c^*$.
  	The remaining $\frac{k-1}{k} \cdot \kapv(c^*)$ units of flow originating at nodes $(v, w)$ are then distributed evenly, via sideways edges of~$w$, among nodes $(v^*, w)$ with $v^* \in V^*$; from the node $(v^*,w)$, the flow can then be routed directly to $(v^*,c^*)$, because $w \succeq_{v^*} c^*$.
  
	Importantly, $V^*$ is large enough: $\card{V^*} \ge \kapv{c^*} / k$ because otherwise, $\comp{V^*}$ $k$-blocks $w$ with witness set $\set{c^*}$.
	Hence, each node $(v^*, w)$ with $v^* \in V^*$ receives at most $k-1$ units of flow in this~last stage of the description of the~$(w, c^*)$-flow $f$. 
	
	Each voter $v$ sends out 1 unit of flow and receives at most $(k-1)+1=k$ units of flow via sideways edges in~$f$. 
	The~incoming flow contributes twice to the cost of $f$ at voter~$v$ since it also gets absorbed at $(v, c^*)$.
	Thus, we showed that $\cost(f) \leq 2k+1$. 
\end{proof}

\section{Tight Lower Bound on $\alpha$-Percentile Distortion}
\label{apx:tight_lower_bound}

\begin{rtheorem}{Theorem}{\ref{thm:perc:deterministic}}
  For every (deterministic) voting rule $\alg$, and constants $\alpha \in [1/2,1)$ and $\epsilon > 0$, 
  there~exists an~election~$\elec$ and candidate $w \in \alg(\elec)$ such that $\dist^\alpha_\elec(w) \geq 5 - \epsilon$.
\end{rtheorem}

\begin{proof}
	Let $m$ solve that $m - \floor{\alpha m} = 2$.
	Such an $m$ exists because the left-hand side is 0 for $m=0$, goes to infinity as $m \to \infty$, and changes by at most 1 going from $m$ to $m+1$.
	
	We construct an election $\elec$ with $m$ candidates and $n=m$ voters.\footnote{The construction can be easily generalized to $n > m$ by replacing each of the $m$ voters with a cluster of $n/m$ voters with identical preferences.}
	Let the candidates be denoted by $c_1, \ldots, c_m$.	
	The voters' preferences are the cyclic shifts of $c_1, c_2, \ldots, c_m$, i.e., for~each~$i$, voter $v_i$ has the ranking 
	$c_i \succ c_{i+1} \succ \cdots \succ c_m \succ c_1 \succ c_2 \succ \cdots \succ c_{i-1}.$
	Because the preferences are completely symmetric, we can assume without loss of generality that ${c_1 \in \alg(\elec)}$.
	
	We construct the following distances consistent with the given rankings, under which $c^* := c_m$ will be much better than $c_1$.
    Let $\delta = \epsilon / (5m)$.
	\begin{align*}
		d(v_1, c_j) & = 10 + \delta \cdot (j-1) \text{ for all } j \\
		d(v_2, c_1) & = 5 \\
		d(v_i, c_j) & = 1 + (j-1) \delta \text{ for all $i \geq 2$ and $j \geq i$} \\
		d(v_i, c_j) & = 3 + (j-1) \delta \text{ for all $i \geq 3$ and $j < i$}.
	\end{align*}
  
	In table form, these distances are as follows:  
	\[\begin{array}{c|c|c|c|c|c|c|c|}
		      & c_1 & c_2       & c_3        & c_4         & \cdots & c_{m-1}        & c_m \\ \hline
		 v_1  & 10  & 10+\delta & 10+2\delta & 10+3\delta & \cdots & 10+(m-2)\delta & 10+(m-1)\delta \\ \hline
		 v_2  & 5   & 1+\delta  & 1+2\delta  & 1+3\delta  & \cdots & 1+(m-2)\delta  & 1+(m-1)\delta \\ \hline 
		 v_3  & 3   & 3+\delta  & 1+2\delta  & 1+3\delta  & \cdots & 1+(m-2)\delta  & 1+(m-1)\delta \\ \hline 
		 v_4  & 3   & 3+\delta  & 3+2\delta  & 1+3\delta  & \cdots & 1+(m-2)\delta  & 1+(m-1)\delta \\ \hline 
		 \vdots & \multicolumn{7}{c|}{\ddots} \\ \hline
		 v_{m-1}  & 3   & 3+\delta  & 3+2\delta  & 3+3\delta  & \cdots & 1+(m-2)\delta  & 1+(m-1)\delta \\ \hline 
		 v_m  & 3   & 3+\delta  & 3+2\delta  & 3+3\delta  & \cdots & 3+(m-2)\delta  & 1+(m-1)\delta \\ \hline
	\end{array}\]
    The fact that the distances are consistent with the preferences in $\elec$ can be seen directly from the~tabular representation.
	Next, we need to verify the triangle inequality, i.e., for every two candidates $c_j, c_{j'}$ and voters $v_i, v_{i'}$, 
	we have that $d(v_i, c_j) \leq d(v_{i'},c_j) + d(v_{i'},c_{j'}) + d(v_i, c_{j'})$. We distinguish several cases:
  	\begin{itemize}
  	\item If $i=1$, then we use that $d(v_i,c_j) - d(v_i,c_{j'}) \leq (m-1) \delta < 1 \leq d(v_{i'},c_j)$.
	\item If $i=2,j=1$, then we use that $d(v_2,c_1) - d(v_{i'},c_1) \leq 2 \leq d(v_2,c_{j'}) + d(v_{i'},c_{j'})$.
	\item If $i \geq 3$ and $j < i$, then $d(v_i,c_j) - d(v_{i'},c_j) \leq (3+(j-1) \delta) - (1+(j-1) \delta) = {2 \leq d(v_i,c_{j'}) + d(v_{i'},c_{j'})}$.
	\item Finally, if $j \geq i > 1$, we use that $d(v_i, c_j) = 1+(j-1) \delta \leq 2 \leq d(v_{i'},c_j) + d(v_{i'},c_{j'})$.
	\end{itemize}
	
	Now, we compare the costs of $c_1$ and $c_m$.
	The choice of $m$ ensures that $\floor{\alpha m} + 1 = m-1$, i.e., both candidates are evaluated based on the second-worst voter.
	All voters except $v_1$ are at distance $1+(m-1) \delta$ from $c_m$, and thus, $\cost^{\alpha}_d(c_m) = 1+(m-1) \delta$.
	On the other hand, $v_1$ is at distance $10$ from $c_1$ and $v_2$ is at distance $5$, so $\cost^{\alpha}_d(v_1) \geq 5$.
	Thus, the $\alpha$-percentile distortion is at least
	\[ \frac{5}{1 + (m-1) \delta} \geq 5 - \epsilon.\]
\end{proof}

\end{document}